\pdfoutput=1
\ifx\synctex\undefined\else\synctex=1\fi
\documentclass[runningheads,a4paper,envcountsame,oribibl]{llncs}

\usepackage{microtype}
\usepackage{amssymb}
\usepackage{mathtools}

\RequirePackage[T1]{fontenc}
\RequirePackage[utf8x]{inputenc}
\RequirePackage{amssymb}
\RequirePackage{stmaryrd}

\DeclareUnicodeCharacter{183}{\cdot}
\DeclareUnicodeCharacter{915}{\Gamma}
\DeclareUnicodeCharacter{916}{\Delta}
\DeclareUnicodeCharacter{918}{\Zeta}
\DeclareUnicodeCharacter{920}{\Theta}
\DeclareUnicodeCharacter{923}{\Lambda}
\DeclareUnicodeCharacter{926}{\Xi}
\DeclareUnicodeCharacter{928}{\Pi}
\DeclareUnicodeCharacter{931}{\Sigma}
\DeclareUnicodeCharacter{933}{\Upsilon}
\DeclareUnicodeCharacter{934}{\Phi}
\DeclareUnicodeCharacter{935}{\Chi}
\DeclareUnicodeCharacter{936}{\Psi}
\DeclareUnicodeCharacter{937}{\Omega}
\DeclareUnicodeCharacter{945}{\alpha}
\DeclareUnicodeCharacter{946}{\beta}
\DeclareUnicodeCharacter{947}{\gamma}
\DeclareUnicodeCharacter{948}{\delta}
\DeclareUnicodeCharacter{949}{\epsilon}
\DeclareUnicodeCharacter{950}{\zeta}
\DeclareUnicodeCharacter{951}{\eta}
\DeclareUnicodeCharacter{952}{\theta}
\DeclareUnicodeCharacter{953}{\iota}
\DeclareUnicodeCharacter{954}{\kappa}
\DeclareUnicodeCharacter{955}{\lambda}
\DeclareUnicodeCharacter{956}{\mu}
\DeclareUnicodeCharacter{957}{\nu}
\DeclareUnicodeCharacter{958}{\xi}
\DeclareUnicodeCharacter{959}{\omicron}
\DeclareUnicodeCharacter{960}{\pi}
\DeclareUnicodeCharacter{961}{\rho}
\DeclareUnicodeCharacter{963}{\sigma}
\DeclareUnicodeCharacter{964}{\tau}
\DeclareUnicodeCharacter{965}{\upsilon}
\DeclareUnicodeCharacter{966}{\phi}
\DeclareUnicodeCharacter{967}{\chi}
\DeclareUnicodeCharacter{968}{\psi}
\DeclareUnicodeCharacter{969}{\omega}
\DeclareUnicodeCharacter{8214}{\parallel}
\DeclareUnicodeCharacter{8614}{\mapsto}
\DeclareUnicodeCharacter{8656}{\Leftarrow}
\DeclareUnicodeCharacter{8657}{\Uparrow}
\DeclareUnicodeCharacter{8658}{\Rightarrow}
\DeclareUnicodeCharacter{8659}{\Downarrow}
\DeclareUnicodeCharacter{8797}{\stackrel{\rm\tiny def}{=}}
\DeclareUnicodeCharacter{8870}{\vdash}
\DeclareUnicodeCharacter{8871}{\models}
\DeclareUnicodeCharacter{9121}{\lceil}
\DeclareUnicodeCharacter{9123}{\lfloor}
\DeclareUnicodeCharacter{9124}{\rceil}
\DeclareUnicodeCharacter{9126}{\rfloor}
\DeclareUnicodeCharacter{9675}{\circ}
\DeclareUnicodeCharacter{10178}{\bot}
\DeclareUnicodeCharacter{10214}{\llbracket} %
\DeclareUnicodeCharacter{10215}{\rrbracket} %
\DeclareUnicodeCharacter{10503}{\Mapsto}    %
\DeclareUnicodeCharacter{10971}{\not\hspace{-0.2em}\cap}
\DeclareUnicodeCharacter{65294}{\ldotp}%
\usepackage{paralist}

\usepackage{enumitem}

\usepackage[pdftex,hidelinks]{hyperref}
\usepackage{listings}
\usepackage{courier}

\usepackage{fancyvrb}
\usepackage{tikz}
\usetikzlibrary{arrows,positioning}

\def\len#1{{\vert{#1}\vert}}
\def\pr{R}
\def\tuple#1{{\left\langle #1 \right\rangle}}
\newcommand{\strangeL}{L_{\angle}}

\title{A Language-theoretic View on Network Protocols}
\author{Pierre Ganty\inst{1} \and Boris Köpf\inst{1} \and Pedro Valero\inst{1,2}}
\institute{IMDEA Software Institute, Madrid, Spain \and Universidad Politécnica de Madrid, Spain \email{\{pierre.ganty,boris.koepf,pedro.valero\}@imdea.org} }

\begin{document}

\pagestyle{headings}  %

\maketitle

\begin{abstract}
Input validation is the first line of defense against malformed or malicious inputs.
It is therefore critical that the validator (which is often part of the parser) is free of bugs.

To build dependable input validators, we propose using parser generators for context-free languages. 
In the context of network protocols, various works have pointed at context-free languages as falling short to specify precisely or concisely common idioms found in protocols. 
We review those assessments and perform a rigorous, language-theoretic analysis of several common protocol idioms.
We then demonstrate the practical value of our findings by developing a modular, robust, and efficient input validator for HTTP relying on context-free grammars and regular expressions.
\end{abstract}
 
\section{Introduction}

Input validation, often carried out during parsing, is the first line of defense against malformed or maliciously crafted inputs.
As the following reports demonstrate, bugs make parsers vulnerable, hence prone to attacks:
a bug in the URL parser enabled attackers to recover user credentials from a widely used password manager~\cite{LastPass},
a bug in the RTF parser led to a vulnerability in Word 2010~\cite{MicrosoftAttack}, 
and a lack of input validation in the Bash shell has been used for privilege escalation by a remote attacker~\cite{HTTPAttack} just to cite a few.
To stop the flow of such reports improved approaches for building input validators are needed.

To build dependable input validators, an approach is to rely   
on mature parsing technologies.
As a candidate, consider parser generators for context-free languages (hereafter CFLs).
Their main qualities are: 
\begin{enumerate}
	\item The code for parsing is synthesized automatically from a grammar specification, shifting the risk of programming errors away from the architect of the validator to the designer of the parser generator;
	\item CFL is the most expressive class of languages supported
          by trustworthy implementation of parser generators. Here by
          trustworthy we mean an implementation that either stood the
          test of time like Flex, Bison, or ANTLR, or that has been
          formally verified like the certified implementations of
          Valiant's~\cite{Bernardy2016} and
          CYK~\cite{Firsov2014} algorithms.
\end{enumerate}
Relying on established CFL technology is an asset compared to existing
solutions which are either programmed from scratch~\cite{HTTPolice} or generated from ad hoc
parser generators~\cite{conf/osdi/BangertZ14,conf/imc/PangPSP06}.

However, even though CFL is the most expressive language class with trustworthy parser generators, previous works suggest CFLs are not enough for network messages. 
Specifically, they claim the impossibility or difficulty to specify
precisely the various idioms found in network messages using CFLs. 
For example, some authors argue that CFL are insufficiently expressive because “data fields that are preceded by their actual length (which is common in several network protocols) cannot be expressed in a context-free grammar”~\cite{conf/esorics/DavidsonSDJ09}.
Yet other authors suggest that going beyond CFLs is merely required
for conciseness of expression, because “it is possible to rewrite
these grammars to […] be context-free, but the resulting specification
is much more awkward”~\cite{conf/ndss/BorisovBWDJG07}.  
Surprisingly, the arguments made are not backed up by any
formalization or proof.

In this paper we formally analyze which idioms can and which cannot be
(concisely) specified using CFLs, and we turn the results into
practice by building an input validator for HTTP messages entirely
based on CFL technology.

As the main contributions of our analysis we find out that:
\begin{compactitem}
\item Length fields of {\em bounded size} are finite and hence form a regular language. 
	However, while they can be concisely represented in terms of a context-free grammar, every finite automaton that recognizes them grows exponentially with the bound. 
	In contrast, length fields of {\em unbounded size} cannot even be expressed as a finite intersection of CFLs.
\item Equality tests between words of {\em bounded size} again form a regular language but, as opposed to length fields, they cannot be compactly represented in terms of a context-free grammar. 
	They do, however, allow for a concise representation in terms of a finite intersection of context-free grammars. 
	Specifically, we show that both the grammar and the number of membership checks grows only linearly with the size bound, which has interesting practical implications, see below. 
	As in the case of length fields, equality checks between words of {\em unbounded size} cannot be expressed as a finite intersection of CFLs.
\end{compactitem}

We consider finite intersections of CFLs because they are
\begin{inparaenum}[(a)]
\item strictly more expressive than CFLs, and
\item checking membership in the intersection of CFLs is equivalent to checking membership in each individual CFL.
\end{inparaenum}

These results lead to a principled and modular approach to input
validation: several CFL parsers are run on the input and their boolean
results (whether the input belongs or not to the CFL) are combined
following a predefined logic to decide whether or not the input
message conforms to the standard (that specifies what valid messages
are).

We demonstrate that this approach is practical by implementing a
proof of concept input validator for a large subset of the HTTP
protocol, covering a significant number of the idioms found in network
messages. Our input validator, called HTTPValidator~\cite{HTTPValidator}
draws inspiration from HTTPolice~\cite{HTTPolice}, a state of the art
input validator for HTTP messages built from scratch by the open
source community.  HTTPValidator
is close to achieve feature parity (in terms of checks) with
HTTPolice and offers competitive performance.

\paragraph{Summary of Contributions.} 
In summary, our contributions are both foundational and applied. On the foundational side we perform a language-theoretic analysis of important protocol idioms, making a step towards more rigor in the field. 
On the applied side we show how to implement an input validator for HTTP using off-the-shelf parser generators.

\paragraph{Paper structure.} 
Section~\ref{sec:prelim} introduces basic language-theoretic and input validation concepts, Section~\ref{sec:ContentLength} discusses the case of length fields, including the chunked messages, whereas Section~\ref{sec:Comparisons} considers the case of comparisons.
We show the practicality of our approach in Section~\ref{sec:practicality}, before concluding with related work (Section~\ref{sec:related}), conclusions and future work (Section~\ref{sec:conclusions}).

\section{Preliminaries}\label{sec:prelim}
\paragraph{Language Theory.}
We begin by introducing the language-theoretic context needed for our
development.  An \emph{alphabet} \(Σ\) is a nonempty finite set of
\emph{symbols}. A \emph{word} \(w\) is a finite sequence of symbols of \(Σ\)
where the empty sequence is denoted \(\varepsilon\).  A \emph{language} is a
set of words and the set of all words over \(Σ\) is denoted \(Σ^*\).  We denote
by \(\len{w}\) the \emph{length} of \(w\).  Further define \( (w)_i \) as the
\(i\)-th symbol of \(w\) if \(1 ≤ i ≤ \len{w}\) and \(\varepsilon\)
otherwise.
Given a nonempty subset \(X\) of \(Σ\) and \(i ∈ \mathbb{N}\) define
\(X^i\) as \(\{ w ∈ X^* \mid \len{w} = i \}\).

We assume the reader is familiar with common operations on languages
such as concatenation and boolean combinations. Likewise, we count on
the reader's familiarity with regular languages and finite-state
automata. Yet we next give a description of context-free grammars,
which are the formal basis of our work.

A \emph{context-free grammar} (or grammar for short) is a tuple $G = (V,
Σ, S, \pr)$ where \(V\) is a finite set of \emph{variables} (or
\emph{non-terminals}) including the \emph{start variable} \(S\); \(Σ\) is
an alphabet (or set of \emph{terminals}), $\pr ⊆ V \times (Σ \cup
V)^{*}$ is a finite set of \emph{rules}. We often write \(X → w\) for
a rule \( (X,w) ∈ \pr\).  We define a \emph{step} as the binary relation
$⇒$ on $(V \cup Σ)^{*}$ given by \(u ⇒ v\) if
there exists a rule \(X → w\) of \(G\) such that \( u = \alpha\, X\,
\beta \) and \( v = \alpha\, w\, \beta\) for some \(\alpha,\beta∈
(V\cupΣ)^*\).  Define \(u ⇒^{*} v\) if there exists a \(n ≥ 0\) steps
sequence \(u_0 ⇒ u_1 ⇒ \ldots ⇒ u_n\) such
that \(u_0 = u\) and \(u_n = v\). A step sequence \(u ⇒^* w\) is
called a \emph{derivation} whenever \(u = S\) and \(w∈ Σ^*\).  Define
\(L(G) = \{ w ∈ Σ^* \mid S ⇒^{*} w \} \) and call it the
language generated by \(G\). A language \(L\) is said to be
\emph{context-free}, or {\em CFL}, if there exists a grammar \(G\) such that \(L =
L(G)\).
The {\em size} of a grammar is the sum of the sizes of its production rules \(R\),
that is, it is given by \(\sum_{(X,w)∈ R} 1{+}\len{w} \).

\paragraph{Input Validation.}\label{sec:DataValidation}

In this paper, validating an input means checking whether it belongs to a language.
In particular, no data structure is filled and no information is extracted from the input other than its membership status.
Thus, validating an input $w$ for a language $L$ means deciding whether \(w\) is a member of \(L\) which is also known as the \emph{membership problem}.
To specify \(L\), we use context-free grammars or regular expressions.

\section{Formal Analysis of Length Fields}\label{sec:ContentLength}

Length fields, whose role is to specify the length of subsequent fields, are commonly found in network protocols such as HTTP~\cite{rfc7230}, SIP~\cite{rfc3261}, DNS~\cite{rfc1035} and UDP~\cite{rfc768}.
As an example, consider the following HTTP POST message:
\begin{center}
	\begin{BVerbatim}[fontsize=\small]
POST /1/notification/list HTTP/1.1\r\n
Content-Length: 47\r\n\r\n
{"header":{},"query":{"count":100},"answer":{}}
	\end{BVerbatim}
\end{center}
The length field begins after the
keyword~\verb!Content-Length:! and terminates before the carriage
return/newline~\verb!\r\n!.
Its content, i.e. 47, describes the length of the message body, which is the string coming after the double~\verb!\r\n!.

In this section, we characterize length fields from the point of view of formal language theory.
We begin with a formalization aiming to capture their essence, and then characterize the class of languages specifying them in the bounded and unbounded cases. 
We consider both cases because some protocols, such as DNS and UDP, require length fields to have fixed size, while others, such as HTTP and SIP, have no such restriction.
We conclude by leveraging these results to analyze chunked transfer encoding.

\subsection{Modeling Length Fields}\label{sec:lengthFields}

To model length fields, we will work with formal languages
over an alphabet $Σ$. For the example of HTTP, \(Σ\)
would be the ASCII character set.
\subsubsection{Fixed Size.}
To describe length fields of finite size $n>0$ we define the
language $L_{\mathit{len}}(n)$ over \(Σ=B\cup W\) as follows:
\[L_{\mathit{len}}(n)\stackrel{\rm\scriptscriptstyle def}{=}\{ x \; w \mid x ∈ B^n, w \in W^*, \len{w} = \textstyle{∑_{i=0}^{n-1}}
(x)_{i+1}\cdot b^{i}\} \]
where $B=\{0,\dots,b-1\}$ for an integer $b>1$.
Intuitively, \(L_{\mathit{len}}(n)\) represents the same number twice, using two different encodings: first \(b\)-ary as \(x\) and then unary as \(w\), where the relationship between both encodings is given by \(\len{w}=∑_{i=0}^{n-1} (x)_{i+1}\cdot b^{i}\).
For example, let \(n=3\), \(B=\{\mathtt{0}, \mathtt{1}\}\) and $W=\{\mathtt{a},\mathtt{b}, \mathtt{c}\}$ the word \(\mathtt{110abc}\) consists of the binary
representation of \(3 = (1\cdot 2^0) + (1\cdot 2^1) + (0 \cdot 2^2) \) followed by a word (\(\mathtt{abc}\)) of length \(3\) and, therefore, $\mathtt{110abc} \in L_{\mathit{len}}(3)$.
We choose this unconventional “least significant digit first” to keep notation simple. The results of this section stay valid for the “most significant digit first” convention.

\subsubsection{Unbounded Size.}
For describing length fields of unbounded size, observe that any overlap between the alphabets $W$ and $B$ for describing the body of the message and its length, respectively, introduces ambiguity as to where the length field ends.
A common approach to remove such ambiguities is to use a \emph{delimiter}, which is a special symbol $\sharp$ not occurring in \(x\) whose aim is to separate explicitly the length field from the body of the message.
We extend the definition of $L_{\mathit{len}}(n)$ to account for such delimiters:
\[L^\sharp_{\mathit{len}}(n)\stackrel{\rm\scriptscriptstyle def}{=}\{ x \; \sharp \; w \mid x ∈ B^n, w \in W^*, \len{w} = \textstyle{∑_{i=0}^{n-1}} (x)_{i+1}\cdot b^{i}\} \enspace .\] 

We are now in position to define a language for describing length fields of arbitrary and unbounded size:
\[\textstyle L^{\sharp}_{\mathit{len}} \stackrel{\rm\scriptscriptstyle def}{=} \bigcup_{i>0} L^\sharp_{\mathit{len}}(i) \enspace .\]

Results shown in this section remain valid when there is no overlap between alphabets $W$ and $B$.
In such case the delimiter is no longer needed and removing it from the results in Section \ref{sec:UnboundedContentLength} has no effect on them.

\subsection{Unbounded Length Fields}\label{sec:UnboundedContentLength}

The following theorem shows that length fields of unbounded size cannot be specified using intersection of finitely many CFLs. 
This means that we need to impose restrictions, such as size bounds, in order to specify length fields using CFLs. We will study fixed size length fields in Section~\ref{sec:FixedSizeContentLength}.

\begin{theorem}\label{theorem:NoIntersectionCFL}
$L^{\sharp}_{\mathit{len}}$ is not a finite intersection of CFLs.
\end{theorem}

To prove this result, we begin by defining the following subset of $L^{\sharp}_{\mathit{len}}$:
\[\strangeL\stackrel{\rm\scriptscriptstyle def}{=}L^{\sharp}_{\mathit{len}} \cap 1^*\sharp^*a^*\enspace .\]

\begin{lemma}\label{lemma:L_cap}
\(\strangeL\) is not a finite intersection of CFLs.
Moreover, no infinite subset of $\strangeL$ is a finite intersection of CFLs.
\end{lemma}

The proof argument relies on semilinear sets which we recall next: a subset of \(\mathbb{N}^k\), with $k>0$, is called \emph{semilinear}, if it can be specified as a finite union of linear sets.
A set \(S ⊆ \mathbb{N}^k\) is called \emph{linear} if there exists \(\vec{b} ∈ \mathbb{N}^k\) and a finite subset \(\{\vec{p}_1,…,\vec{p}_m\}\) of \(\mathbb{N}^k\) such that
\noindent\hfill$S = \{ \vec{b} + λ_1\, \vec{p}_1 + ⋯ + λ_m\, \vec{p}_m \mid λ_1,…,λ_m ∈ \mathbb{N} \} \enspace .$\hfill\mbox{}\par

Let $\bar{w}=\tuple{w_1,…,w_k}$ be a tuple of \(k>0\) words, define a mapping $f_{\bar{w}}\colon\mathbb{N}^k\to w_1^*… w_k^*$ by
$f_{\bar{w}}(i_1,…,i_k)=w_1^{i_1} … w_k^{i_k}$, that is, the output of $f_{\bar{w}}$
is a word in which the $i$-th component of $\bar{w}$ is repeated a number of
times that corresponds to the $i$-th input to $f_{\bar{w}}$. We define the
preimage of $f_{\bar{w}}$ and liftings of $f_{\bar{w}}$ from elements to subsets of
$\mathbb{N}^k$ in the natural way.

The following result by Latteux~\cite{Latteux_1979} establishes a fundamental correspondence between languages given by finite intersection of CFLs and semilinear sets.
\begin{proposition}[{\cite[Prop 7]{Latteux_1979}}]\label{prop:equivalents}
  Let $\bar{w}=\tuple{w_1,…,w_k}$, \(k>0\),
  and \(L ⊆ w_1^* … w_k^*\):
  $f_{\bar{w}}^{-1}(L)$ is semilinear if and only if $L$ is a finite intersection of CFLs.
\end{proposition}

Now we meet the requirements to prove Lemma \ref{lemma:L_cap}.
\begin{proof}[Sketch]
The proof of Lemma~\ref{lemma:L_cap} relies on the observation that
\[\strangeL = \{ 1^n \; \sharp \; a^{\mathit{val}}
\mid  \mathit{val} = \textstyle{∑_{i=0}^{n-1}} b^{i}\} \enspace .\]
Let $\bar{w}=\tuple{1,\sharp,a}$, since \( ∑_{i=0}^{n-1}  b^{i} = \frac{b^{n}-1}{b-1}
\) for all \(b>1\), we obtain:
\begin{equation}\label{eq:LinearSetFormula}
	f_{\bar{w}}^{-1}(\strangeL) = \left\{\bigl(i,1,\textstyle{\frac{b^{i}-1}{b-1}}\bigr) \mid i\ ∈ \mathbb{N} \right\}\enspace .
\end{equation}

Next, we show this set is not semilinear using the facts that
\begin{inparaenum}[(a)]
\item the third component grows exponentially in $i$, and
\item $f_{\bar{w}}^{-1}(\strangeL)$ is an infinite set.
\end{inparaenum}
The definition of semilinear set then shows that by taking two elements in (b) we can obtain a third one.
We then show that those three elements violate (a) unless they all coincide.
The same reasoning remains valid when considering an infinite subset of \(\strangeL\).
The details of the proof are given in Appendix~\ref{sec:proofs}.
\qed
\end{proof}

Once Lemma \ref{lemma:L_cap} is proved, the proof of Theorem \ref{theorem:NoIntersectionCFL} easily follows. 
\begin{proof}[of Theorem~\ref{theorem:NoIntersectionCFL}]
Assume $L^{\sharp}_{\mathit{len}}$ is a finite intersection of CFLs.
Since \(1^*\sharp^* a^*\) is a CFL, $\strangeL$ is also a finite intersection of CFLs contradicting Lemma~\ref{lemma:L_cap}.
\qed
\end{proof}

Our definitions of $L_{\mathit{len}}(n)$ and $L^{\sharp}_{\mathit{len}}$ do not put any constraints on the structure of the word \(w\) that follows the length field and the delimiter (if any).
In practice, however, the word \(w\) may need to satisfy constraints beyond those on its length, such as containment in a specific language.

\begin{theorem}\label{theorem:constraint}
  The language $L^{\sharp}_{\mathit{len}} \cap \{ x \, \sharp\, w \mid x \in B^*,\; w \in L_C \}$ is a finite intersection of CFLs for \textbf{no} infinite CFL \(L_C \subseteq W^*\).
\end{theorem}

The proof of this Theorem follows the same argument used to prove Theorem \ref{theorem:NoIntersectionCFL}.
Hence, we begin by defining a subset of $L^{\sharp}_{\mathit{len}} \cap \{ x \, \sharp\, w \mid x \in B^*,\; w \in L_C \}$ for which Proposition \ref{prop:equivalents} holds.

Let $S$ be the start symbol of the grammar generating language $L_C$.
Since language $L_C$ is infinite the following must hold: for some non terminal $A$ and $α_i \in W^*$, we have $S  ⇒^* α_1 A α_5; \;  A  ⇒^* α_2 A α_4; \;  A ⇒^* α_3$ with $α_2 \neq ε$ or $α_4 \neq ε$.

It follows that $\{α_1α_2^iα_3α_4^iα_5 \mid i \geq 0\} \subseteq L_C$ and, thus, 
\[L_{\rightthreetimes}\stackrel{\rm\scriptscriptstyle def}{=}L^{\sharp}_{\mathit{len}} \cap \{ x \, \sharp\, w \mid x \in B^*,\; w \in L_C \} \cap 1^*\sharp^*α_1^*α_2^*α_3^*α_4^*α_5^*\]
is an infinite language contained in $1^*\sharp^*α_1^*α_2^*α_3^*α_4^*α_5^*$.

\begin{lemma}\label{lemma:constraint}
Language $L_{\rightthreetimes}$ is not a finite intersection of CFLs. Moreover, no infinite subset of $L_{\rightthreetimes}$ is a finite intersection of CFLs.
\end{lemma} 
This Lemma is similar to Lemma \ref{lemma:L_cap} and so is the proof, whose details are given in Appendix \ref{sec:proofs}.
Finally, we proceed to prove Theorem \ref{theorem:constraint} by contradiction.
\begin{proof}[of Theorem~\ref{theorem:constraint}]
Assume that $L^{\sharp}_{\mathit{len}} \cap \{ x \, \sharp\, w \mid x \in B^*,\; w \in L_C \}$ is a finite intersection of CFLs.
Since $1^*\sharp^*α_1^* α_2^*α_3^*α_4^*α_5^*$ is context-free, $L_{\rightthreetimes}$ is also a finite intersection of CFLs, which contradicts Lemma \ref{lemma:constraint}.
\qed
\end{proof}

\subsection{Fixed Size Length Fields}\label{sec:FixedSizeContentLength}

In this section we sidestep the negative results of Section~\ref{sec:UnboundedContentLength} by assuming an upper bound on the length field which indeed occurs in some network protocols.
Such is the case of the IP, UDP and DNS protocols, whose specifications~\cite{rfc791,rfc1035,rfc768} define 16-bit fields containing the length of the data in terms of bytes.
In some cases, assuming an upper bound on the length field, even if it is not defined by the standard, yields no loss of generality for all practical purposes. 
It is the case for HTTP where the majority of implementations do assume a bound on the size of length fields (e.g. major web browsers all do).

We start with the family of languages 
$L_{\mathit{len}}(n)$ where the length field is \(n\) symbols long. 
It is easy to see that each language of this family is finite, hence regular.
Now we turn to the size of specifications for \(L_{\mathit{len}}(n)\).
In terms of finite state automata, all automata specifying \(L_{\mathit{len}}(n)\) grow exponentially in \(n\).
Let $b>1$ be the base in which the length is encoded, then there are $b^n$ possible encodings for the length.
By the pigeonhole principle, having less than $b^n$ reachable states after reading the first \(n\) symbols implies that two distinct length encodings end up in the same state, making them indistinguishable for the automaton.
Hence, it cannot decide \(L_{\mathit{len}}(n)\).
However, when \(L_{\mathit{len}}(n)\) is specified using context-free grammars, we show that it admits a more compact description.

\begin{theorem}\label{theorem:compactboundedcontentlength}
Let $Σ$ be fixed alphabet and \(n>0\), there exists a context-free grammar \(G_{\mathit{len}}(n)\) of size \(\mathcal{O}(n)\) such that \(L(G_{\mathit{len}}(n))=L_{\mathit{len}}(n)\).
\end{theorem}
\begin{proof}
For simplicity of presentation we assume that length fields are encoded in binary, i.e. $b=2$ in the definition of $L_{\mathit{len}}(n)$.
The generalization to any \(b>2\) is tedious but straightforward.

The grammar \(G_{\mathit{len}}(n)\) is defined by the alphabet \(Σ\), variables \(\{S\}\cup \{X_i \mid 0\leq i\leq n\} \cup \{F_i \mid 0 \leq i \leq n-1 \} \) and the following rules:
\begin{align*}
	\{ S & → X_{0} \} &	\{	X_{n} & → \varepsilon \}  \\
	\{ X_i & → 0\; X_{i+1} \mid 0 ≤ i < n\} & \{ X_i & → 1\; X_{i+1}\; F_{i} \mid 0 ≤ i < n\} \\
	\{ F_{j} & → F_{j-1}\; F_{j-1} \mid 1≤ j ≤ n-1\} & \{ F_0 & → c \mid c ∈ W \}
\end{align*}

It follows by construction that $L(G_{\mathit{len}}(n))=L_{\mathit{len}}(n)$.
A closer look reveals that, since the alphabet is fixed and therefore so is $\len{Σ} \geq \len{W}$, the size of the rules of each set is bounded and independent from $n$ while there are $3n{+}2{+}\len{W}$ rules so the size of \(G_{\mathit{len}}(n)\)
is $\mathcal{O}(n)$.
\qed
\end{proof}

Next, we show that \(110abc ∈ L_{\mathit{len}}(3)\) is also contained in \(L(G_{\mathit{len}}(3))\).
\begin{multline*}
S ⇒ X_0 ⇒ 1 X_1 F_0 ⇒ 1 1 X_2 F_1 F_0 ⇒ 1 1 0 X_3 F_1 F_0 \\
⇒ 1 1 0 F_1 F_0 ⇒ 1 1 0 F_0 F_0 F_0 ⇒^* 1 1 0 abc
\end{multline*}

\subsection{Chunked Messages}\label{sec:chunkMessages}

Closely related to length fields are chunked messages, a feature found in the HTTP protocol.
According to the standard, the header \verb!Transfer-Encoding: chunked! signals that the body of the message is divided into chunks, each of which has its size defined by a variable size length field as shown below:
\begin{center}
\begin{BVerbatim}[fontsize=\footnotesize]
HTTP/1.1 200 OK\r\n
Transfer-Encoding: chunked\r\n\r\n
12\r\nThe file is \r\n
16\r\n3,400 bytes long\r\n
0\r\n\r\n
	\end{BVerbatim}
\end{center}

Relying on previous definitions we model chunked messages by defining the languages $L_{\mathit{chunk}}^{\sharp} \stackrel{\rm\scriptscriptstyle def}{=} \bigl(L^{\sharp}_{\mathit{len}}\;\{\sharp\}\bigr)^+$ and $L_{\mathit{chunk}}^{\sharp}(n) \stackrel{\rm\scriptscriptstyle def}{=}\bigl(L_{\mathit{len}}(n)\;\{\sharp\} \bigr)^+$ for unbounded and fixed (given by \(n\)) length field size, respectively.
We further assume $\sharp \notin W$ and $Σ=B\cup W \cup \{\sharp\}$ to recognize the end of each chunk and thus avoid ambiguity.

Next, we turn to the claims found in the literature~\cite{conf/esorics/DavidsonSDJ09} about the impossibility of specifying chunked messages using CFLs.
The proofs, which are slight variations of proofs for length fields, can be found in Appendix~\ref{sec:proofs}.

\begin{theorem}\label{theorem:chunked}
$L_{\mathit{chunk}}^{\sharp}$ is not a finite intersection of CFLs.
\end{theorem}
\begin{theorem}\label{theorem:economychunked}
Let $Σ$ be a fixed alphabet and \(n>0\). The language $L_{\mathit{chunk}}^{\sharp}(n)$ is regular and can be specified by a context-free grammar of size \(\mathcal{O}(n)\).
\end{theorem}

\section{Formal Analysis of (In)equalities}\label{sec:Comparisons}
Input validation sometimes requires comparing different parts of a
message, e.g., to check that two subwords are identical
or that the first one represents a number or a date that is greater
than the second one. 
For instance, an HTTP GET message is valid only if the field \verb!last-byte-pos! is greater than \verb!first-byte-pos!.

\subsection{Equality Check}\label{sec:equalitychecks}
Consider the case of HTTP when a client is asking for a transition to some other protocol.
As the standard of mandates, equality should hold between the \verb!Upgrade! fields of the request and its response.

\begin{center}
\begin{BVerbatim}[fontsize=\footnotesize]
======== REQUEST ========   ======== RESPONSE ========
GET /example HTTP/1.1\r\n   HTTP/1.1 101 Switching Protocols\r\n
Upgrade: h2c\r\n            Connection: Upgrade\r\n
                            Upgrade: h2c\r\n
\end{BVerbatim}
\end{center}

\subsubsection{Modeling Equality Check.}
We begin our study of comparisons with the case of two contiguous subwords compared for equality.
To this end consider the following language over alphabet \(Σ\) given by
\[ L_{=}^{\sharp} \stackrel{\rm\scriptscriptstyle def}{=} \{x\, \sharp\, y \mid x=y \} \enspace . \]
This language consists of twice the same word with ‘\(\sharp\)’ in between.
Again, we assume \(\sharp\) occurs in \(x\) for no \(x\).

When the size of the words \(x\) and \(y\) is fixed, the delimiter is no longer needed.
Thus, we define \(L_=(n)\)
\[L_=(n) \stackrel{\rm\scriptscriptstyle def}{=} \{x \; y \mid x,y \in Σ^n \; \land \; x=y\} \enspace .\]

\subsubsection{Unbounded Size.}
We now consider the case where the length of the subwords to compare is unbounded.
The example at the top of the section requires, when validating a request-response pair of HTTP messages, to check equality across \verb!Upgrade! fields.

This situation is described by the language $L_{=}^{\sharp}$.
Next, we recall results by Liu and Weiner \cite{Liu1973} and Brough \cite{brough2010groups} enabling us to show that $L_{=}^{\sharp}$ is not a finite intersection of context-free languages.

\begin{proposition}[{\cite[Prop 2.1]{brough2010groups}}]\label{prop:closegsm}
For every $k>0$, the set of languages that are an intersection of $k$ CFLs is closed under
\begin{inparaenum}[\upshape(\itshape i\upshape)]
\item	inverse GSM mappings, and
\item union with context-free languages.
\end{inparaenum}
\end{proposition}

\begin{theorem}[{\cite[Thm 8]{Liu1973}}]\label{theorem:Km1Intersection}
	Let $a_1,\ldots a_k$ be \(k>0\) distinct symbols. Then
$L_{(k)}\stackrel{\rm\scriptscriptstyle def}{=}\{a_1^{i_1}a_2^{i_2}\ldots a_k^{i_k}a_1^{i_1}a_2^{i_2}\ldots
 a_k^{i_k} \mid  i_j \geq 0 \text{ for all } j\}$ is an intersection of \(\ell\) context-free languages for \textbf{no} \(\ell<k\).
\end{theorem}

We are now in position to prove our impossibility result about \(L_{=}^{\sharp}\).

\begin{theorem}\label{theorem:Lww}
$L_=^{\sharp}$ is not a finite intersection of CFLs.
\end{theorem}
\begin{proof}[Sketch.]
For the proof sketch we deliberately ignore the delimiter.
Details about how to deal with it can be found in Appendix \ref{sec:proofs}.

Assume $L_{=}$ (the delimiterless version of \(L_{=}^{\sharp}\)) is an intersection of $m$ CFLs.
Now observe that $L_{(k)}=L_{=}\cap a_1^*a_2^*\ldots a_k^* a_1^*a_2^*\ldots a_k^* $
This implies \(L_{(k)}\) is an intersection of $m+1$ context-free languages, which contradicts Theorem~\ref{theorem:Km1Intersection} for $k>m+1$.
\qed
\end{proof}

\subsubsection{Fixed Size.}
Because of the negative result of Theorem~\ref{theorem:Lww} we turn back again to the restriction assuming an upper bound on the length of the subwords to compare. 
We argue next that, in practice, such a restriction is reasonable. 

Consider the following HTTP message. 
\begin{center}
\begin{BVerbatim}[fontsize=\small]
HTTP/1.1 200 OK\r\n
Date: Sat, 25 Aug 2012 23:34:45 GMT\r\n
Warning: 112 - "Net down" "Sat, 25 Aug 2012 23:34:45 GMT"\r\n
\end{BVerbatim}
\end{center}
The RFC mandates that the date in the \verb!Warning! header be equal to \verb!Date!.
Since date formats have bounded length we immediately have an upper bound of the length of the subwords to compare. 

Another example is given by the MIME protocol which allows to split messages into multiple parts provided they are flanked by a user-defined delimiter string. Let us consider an example:
\begin{center}
\begin{BVerbatim}[fontsize=\small]
MIME-Version: 1.0\r\n
Content-type: multipart/mixed; boundary="Mydelimiter"\r\n\r\n
PREAMBLE to be ignored\r\n--Mydelimiter\r\n
Plain ASCII text.\r\n--Mydelimiter\r\n
Plain ASCII text.\r\n--Mydelimiter--\r\n
EPILOGUE to be ignored.\r\n
\end{BVerbatim}
\end{center}
Observe that the delimiter is first declared, \verb+boundary="Mydelimiter"+, and then \verb+Mydelimiter+ is used three times, the first two times as \verb+--Mydelimiter+ the last time as \verb+--Mydelimiter--+.

Equality checks can ensure each part is flanked with the same delimiter.
In the case of MIME, the standard~\cite{rfc1341} imposes a maximum length of 69 symbols for the delimiter giving us an upper bound.

Equality checks for a fixed number $n$ of symbols are specified by $L_{=}(n)$. 
For every \(n\), the language $L_{=}(n)$ is finite, hence regular.
Nonetheless Theorem~\ref{theorem:Language_square_words}, due to Filmus~\cite{CFLBound}, states that this language has no “compact” specification as a grammar.\footnote{This implies it has no “compact” specification by a finite state automaton either.}
Still it can be represented “compactly” as a finite intersection of CFLs, as shown by Theorem~\ref{theorem:Lwwn}.
In this section we will study the size of different grammars assuming the alphabet $Σ$ is fixed and, thus, $\len{Σ}$ is a constant.

\begin{theorem}[{\cite[Thm 7]{CFLBound}}]\label{theorem:Language_square_words}
Let $\len{Σ}>2$, every context-free grammar for $L_{=}(n)$ has size

\[Ω\left(\frac{\len{Σ}^{n/4}}{\sqrt{2n}} \right) \enspace .\]

\end{theorem}
Recall that $f(n)=\Omega(g(n))$ means that $f$ is bounded from below\footnote{$f(n)=Ω(g(n)) \text{ if{}f } \exists \text{ positive } c,n_0 \text{ s.t. } \forall n>n_0, \;f(n) \geq c\cdot g(n)$} by $g$ for sufficiently large $n$, which implies that context-free grammars for $L_{=}(n)$ exhibit exponential growth in $n$.

Our next theorem based on the observation that \(x=y\) if{}f
\( (x)_i = (y)_i \) for all \(i\) allows to capture \(L_{=}(n)\) as
a intersection of \(n\) CFLs.
\begin{theorem}\label{theorem:Lwwn}
	Let the alphabet \(Σ\) be fixed, the language
$L_{=}(n)$ is an intersection of $n$ CFLs, each of which is specified by a grammar of size $\mathcal{O}(n)$.
\end{theorem}
\begin{proof}
	Given \(i\in \{1,…,n\}\), define the language \(L_{=_i}(n)\) over the alphabet \(Σ\) given by
\[ L_{=_i}(n) \stackrel{\rm\scriptscriptstyle def}{=} \{x\, y \mid x,y ∈ Σ^n,\; (x)_i=(y)_i\}\enspace .\]

Clearly, for every word \(u\) we have \(u\in L_{=}(n)\) if{}f \(u \in L_{=_i}(n)\) for all \(i\in\{1,…,n\}\).
Next, define $G_{=_i}$ as the grammar for \(L_{=_i}(n)\) with variables \(S\) and \(T\), alphabet \(Σ\) and the rules:$\{ S  → T^{i-1}\; c\; T^{n-1}\; c\; T^{n-i} \mid c \in Σ\}, \;\;\; \{ T → c \mid c ∈ Σ \}$.
It is routine to check that the size of the grammar is $\mathcal{O}(n)$.
\qed
\end{proof}

Above, we studied specification of equality checks for two contiguous subwords.
In practice, however, comparisons are often more general.
In the previous HTTP example, the dates to compare for equality are not necessarily contiguous.
Also, to specify the split messages of MIME using equality checks we need to generalize to the cases where equality covers more than two subwords (each of the multiple parts is flanked with the same delimiter) and those are not necessarily contiguous (some parts are non empty). 

We show this generalization of equality checks can still be specified concisely by a finite intersection of CFLs. 
For the sake of space, results are deferred to Appendix~\ref{sec:dynamicDelimitersAppendix}.

\subsection{Inequality Checks}

Thus far, we have focused on languages whose words consist of two equal subwords.
However, comparisons sometimes require that the first subword represents a lower number than the second or an earlier date.
The following request is asking for bytes of \verb!BigBuckBunny.mp4! between offsets \verb!2833! and \verb!7026!.
To be valid the requested range should describe a non empty set.
\begin{center}
\begin{BVerbatim}[fontsize=\footnotesize]
GET /BigBuckBunny.mp4\r\n
Range: bytes=2833-7026\r\n
\end{BVerbatim}
\end{center}

\subsubsection{Modeling Inequality Checks.}
Let $\preceq$ define a total order on $Σ$.
We extend \(\preceq\) to $Σ^*$ and denote it \(\preceq^*\) as follows.
We first define \(\preceq^*\) when its arguments have equal length, then we proceed with the general case.

Given \(x, y ∈ Σ^* \) of equal length, let \(p\) be the least position such that \( (x)_p \neq (y)_p\).
Then \( (x)_p \preceq (y)_p \) if{}f \(x \preceq^* y\).
Otherwise (no such position \(p\) exists) we also have \(x \preceq^* y\) since the two words are equal.

Let us now proceed to the case where \(x\) and \(y\) have different length and assume \(x\) is the shortest of the two words (the other case is treated similarly).
Then we have \(x \preceq^* y\) if{}f \( x' \preceq^* y \) where \( x' = \min_{\preceq}(Σ)^{\len{y}-\len{x}} x \), that is \(x'\) is the result of padding \(x\) with the minimal element of \(Σ\) so that the resulting word and \(y\) have equal length.
For instance, \(5 \preceq^* 21 \) because \( 05 \preceq^* 21 \) where \(Σ \) is the set of all digits and \(\preceq\) is defined as expected.
It is an easy exercise to check that \(\preceq^{*}\) is a total order (hint: \(\preceq\) is a total order).

\subsubsection{Unbounded Size.}

Let us turn back to the \verb!Range! field example at the top of the section.
To specify the language of valid ranges, since the two subwords are unbounded, a delimiter is needed to indicate the end of the first word.
In our example the delimiter is the dash symbol.

Next we define $L_{\preceq}^{\sharp}$, the language deciding unbounded size inequality check using \(\sharp\) as a delimiter, as follows:
\[L_{\preceq}^{\sharp} \stackrel{\rm\scriptscriptstyle def}{=} \{ x \;\sharp\; y \mid x,y \in Σ^*, \; x \preceq^* y\}\enspace .\]

\begin{theorem}
The language $L_{\preceq}^{\sharp}$ is not a finite intersection of CFLs.
	\label{th:preceqnotcfl}
\end{theorem}
\begin{proof}
We begin by defining the order $\succeq$ over $Σ$ as \(\preceq^{-1}\) and define $\succeq^*$ by replacing \(\preceq\) with \(\succeq\) in the definition of \(\preceq^*\).
Clearly \(a \preceq b\) if{}f \(b \succeq a\) holds, hence there exists a permutation \(γ\) on \(Σ\).
Indeed, we can write \(Σ\) as the set \( \{a_1,…,a_n\} \) such that \(a_i \preceq a_j\) if{}f \(i \leq j\).
Now define \(γ \colon a_i \mapsto a_{n{+}1{-}i}\).
It follows that \( a \preceq b\) if{}f \(\gamma(a) \succeq \gamma(b)\).
The previous equivalence naturally lifts to words (\(\preceq^*\) and \(\succeq^*\)), e.g. \( v \preceq^* w\) if{}f \(\gamma(v) \succeq^* \gamma(w)\).

Next define
\[L_{\succeq}^{\sharp} \stackrel{\rm\scriptscriptstyle def}{=} \{ x \; \sharp \; y \mid x,y \in Σ^*, \; x \succeq^* y\}\enspace .\]
Notice that the following equality holds: \(L_{\succeq}^{\sharp} = \{\gamma(x)\; \sharp\; \gamma(y) \mid x\;\sharp\; y\in L_{\preceq}^{\sharp} \} \).
Stated equivalently, \(\gamma(L_{\preceq}^{\sharp}) = L_{\succeq}^{\sharp}\) where \(\gamma\) is lifted to be a language homomorphism and also \(L_{\preceq}^{\sharp} = \gamma^{-1}(L_{\succeq}^{\sharp})\) since \(\gamma\) is a bijection.

Following Proposition~\ref{prop:closegsm} \upshape(\itshape i\upshape) finite intersections of CFLs are closed under inverse GSM mapping.
This implies that they are also closed under inverse homomorphism such as \(\gamma^{-1}\).

Assume $L_{\preceq}^{\sharp}$ is a finite intersection of CFLs.
It follows from above that so is $L_{\succeq}^{\sharp}$.
Finally, consider the equivalence \( v = w\) if{}f \(v\preceq^* w\) and \(v \succeq^* w\).
Lifted to the languages the previous equivalence becomes: $L_=^{\sharp} = L_{\preceq}^{\sharp} \cap L_{\succeq}^{\sharp}$.

Since both $L_{\preceq}^{\sharp}$ and $L_{\succeq}^{\sharp}$ are finite intersection of CFLs we conclude that so is \(L_=^{\sharp}\) which
contradicts Theorem~\ref{theorem:Lww}.
\qed
\end{proof}

\subsubsection{Fixed Size.}
With the same motives as for equality checks we turn to the case in which the size of the words to be compared is fixed, say $n$.
As opposed to the unbounded case, we can discard the delimiter
because \(n\) -- the last position of the first word -- is known.
The message below illustrates an inequality check between fixed size subwords.
\begin{center}
\begin{BVerbatim}[fontsize=\small]
HTTP/1.1 304 Not Modified\r\n
Date: Tue, 29 Mar 2016 09:05:57 GMT\r\n
Last-Modified: Wed, 24 Feb 2016 15:23:38 GMT\r\n
\end{BVerbatim}
\end{center}

To ensure that this response is valid the \verb!Last-Modified! field must contain a date earlier than the \verb!Date! field.

Let \(n > 0 \), we define $L_{\preceq}(n)$ to be:
\[ L_{\preceq}(n) \stackrel{\rm\scriptscriptstyle def}{=} \{x\;y \mid x,y \in Σ^n, \; x \preceq^* y\}\enspace .\]

\begin{theorem}\label{theorem:Lleq}
	Let the alphabet \(Σ\) be fixed and \(n>0\), $L_{\preceq}(n)$ is a boolean combination of $2n$ languages each one specified by a grammar of size $\mathcal{O}(n)$.
\end{theorem}
\begin{proof}
Let $G_{=_i}$ be the grammars used in the proof of Theorem \ref{theorem:Lwwn} and let $G_{\preceq_i}$ a grammar for the language $L_{\preceq_i}(n)\stackrel{\rm\scriptscriptstyle def}{=}\{x\;y \mid x,y \in Σ^n, \; (x)_i \preceq (y)_i\}$.
Then, by definition of the order $\preceq$ over $Σ^n$, we write
\[
	w ∈ L_{\preceq}(n) \Leftrightarrow w ∈ L_{=_{1..n}}(n)
	\bigvee_{i=1}^n \left( w ∈ L_{=_{1..i-1}}(n) \land w \notin L_{=_i}(n) \land w ∈ L_{\preceq_i}(n) \right)
\]
where \(w ∈ L_{=_{1..i}}(n)\) is equivalent to \(w ∈ \bigcap_{j=1}^i L_{=_j}(n)\).

The size of each grammar $G_{=_i}$ was shown to be $\mathcal{O}(n)$.
On the other hand, each grammar $G_{\preceq_i}$ is defined by the alphabet \(Σ\), variables $S$, $T$ and $\{T_a\mid a \in Σ\}$ and the rules:
\[
\{ S → T^{i-1}\; a\; T^{n-1}\; T_a\; T^{n-i} \mid a \in Σ\} \quad \{ T_a → c \mid
c ∈ Σ, a \preceq c \} \quad \{ T → c \mid c ∈ Σ \} \enspace .
\]
It is routine to check that the size of the grammar is $\mathcal{O}(n)$.
\qed
\end{proof}

The language $L_{\preceq}(n)$ can be extended to describe the situation in which $x$ and $y$ represent dates and $\preceq$ means “earlier than”.
To this end, whenever the month is given by its name instead of the number thereof we should read it as a single symbol, considering each one as an element of the alphabet.
Otherwise, a comparison between numbers as described in proof of Theorem \ref{theorem:Lleq} will work.
Once we know how to compare the years, months, and days of two dates, combining them to construct the language comparing two dates is straightforward.

\section{Practical Evaluation}\label{sec:practicality}

The results given in Sections~\ref{sec:ContentLength} and
\ref{sec:Comparisons} characterize the extent to which (intersections
of) CFL can be used to specify common idioms of network protocols.  In
this section, we demonstrate that the positive theoretical results can
be turned into practical input validators for real-world network
protocols. We begin by discussing practical encoding issues, before we
present an input validator for HTTP.

\subsection{Encoding Real-World Protocols as CFG}

\subsubsection{Encoding effort.}
The manual effort of translating protocol specification into grammars is facilitated by the RFC format:
Protocol RFCs usually consist of a grammar accompanied by a list of additional constraints written in English.
This grammar is typically given in ABNF format~\cite{rfc5234} which easily translates to a context-free grammar.
The additional constraints translate to regular expressions or CFGs,
along the lines described in this paper.
Then the set of valid messages of the protocol is described by a boolean combination of small CFLs.  

\subsubsection{Encoding size.}
The grammars required to perform the validation against the idioms
discussed in this paper remain small even for real-world protocols:

\paragraph{Length Fields.}
The CFG for $L_{\mathit{len}}(n)$ consists of \(3n+2\) rules, i.e. it
grows linearly in the size of the length field. This implies that it
grows only {\em logarithmically} with the size of the message body,
which makes the CFG encoding practical for real-world scenarios.
\paragraph{Comparisons.}
To compare two strings of length \(n\) we need \(2n\) grammars each with no more than \(3 \len{\Sigma}\) rules where \(\len{\Sigma}\) is the size of the alphabet.
In practice, \(n\) is small because it is the length of the encoding of a position within a file, a timestamp, a hash value,…

\subsection{An Input Validator for HTTP}

Next we report on HTTPValidator~\cite{HTTPValidator}, a proof of concept implementation to validate HTTP messages based on mere CFGs and regular expressions, without using attributes nor semantic actions.

\paragraph{Why HTTP?}
First, HTTP contains almost all of the features that have been used in the literature \cite{conf/ndss/BorisovBWDJG07,conf/imc/PangPSP06} to dispute the suitability of CFLs for parsing network protocols.  
Second, HTTP is a widely used and complex protocol, making it an ideal testbed for our approach.  
Finally, HTTPolice~\cite{HTTPolice} is a lint for HTTP messages which checks them for conformance to standards and best practices and provides a reference for comparison.

\paragraph{HTTP as CFG.}
The ABNF described by the standard~\cite{rfc7230} is translated into a single CFG while constraints such as “A client MUST send a Host header field in all HTTP/1.1 request messages.” and “A client MUST NOT send the chunked transfer coding name in TE” are translated into regular expressions and CFGs.

\paragraph{Implementation.}
Regular expressions and grammars are compiled with Flex and Bison respectively.
We avoid conflicts altogether by relying on the \verb!%glr-parser!  declaration, which forces Bison to produce a generalized LR parser \footnote{\url{https://www.gnu.org/software/bison/manual/html_node/GLR-Parsers.html}} that copes with unresolved conflicts without altering the specified language.
Finally, a script runs all these validators sequentially and combines their boolean outputs to conclude the validation.
Table~\ref{table:sizes} describes the sizes of each separate element of our validator.
Further details can be found in the repository~\cite{HTTPValidator}.
\begin{table}[ht]
\caption{Sizes of the formal languages required to validate an HTTP message}
\label{table:sizes}
\centering
\begin{tabular}{lcl}
\hline
\textbf{Feature} & &\textbf{Size} \\
\hline
HTTP ABNF as a CFG & & 1013 grammar rules\\
Decimal length field of size up to 80 & & 871 grammar rules \\
Comparison of version numbers & & 3 grammars with 13 grammar rules each\\
Constraints (91 different ones) & & 260 regular expressions \\
\hline
\end{tabular}
\end{table}

\paragraph{Evaluation.}
We evaluate HTTPValidator on messages obtained from real-world traffic (using Wireshark) and on messages provided with HTTPolice as test cases. 
In total we thus obtain 239 test cases of which HTTPolice classifies 116 as valid and 123 as invalid HTTP. 
We run HTTPValidator on these test cases obtaining the same classification as HTTPolice but for two false positives.
These errors are due to well-formedness checks on message bodies in JSON and XML format, which we currently do not consider in HTTPValidator but HTTPolice does.

The time required for evaluating all test cases \footnote{We run our experiments on an Intel Core i5-5200U CPU 2.20GHz with 8GB RAM.} is 16.1s for HTTPValidator and 60s for HTTPolice, i.e.\ we achieve a 4-fold speedup. 
Note that this comparison is slightly biased towards HTTPValidator because HTTPolice relies on interpreted Python code whereas the parsers in HTTValidator are compiled to native code. 
Moreover, we store each of the test cases in a single file, forgoing HTTPolice's ability to process several HTTP messages in a single file.
On the other hand, we have put ease of implementation before performance so no parallelization has been implemented so far.

Overall, the experimental evaluation shows that, on our testbed, HTTPValidator achieves coverage and performance that is competitive with the state-of-the-art in the field, thereby demonstrating the practicality of our approach.

\section{Related Work}\label{sec:related}

We discussed related work on language theory and input validation in
the paper body. Here we focus on discussing recent efforts for
building parser generators for network protocols.

In recent years, a number of parser generators for network protocols
have emerged. They are often parts of larger projects, but can be used
in a stand-alone fashion. Important representatives are
BinPac~\cite{conf/imc/PangPSP06}, which is part of the Bro Network
Security Monitor~\footnote{\url{https://www.bro.org/}},
UltraPac~\cite{conf/sigcomm/LiXGTCLJL10}, which is part of the
NetShield Monitor, Gapa \cite{conf/ndss/BorisovBWDJG07},
FlowSifter~\cite{conf/infocom/MeinersNLT12}, and Nail
\cite{conf/osdi/BangertZ14}. The difference to our approach is that
they are all are built from scratch, whereas we rely on established
CFG parsing technology. Moreover, they rely on user-provided code for
parsing idioms such as length fields, whereas we specify everything in
terms of (intersections of) CFG. However, we emphasize that the focus
of our approach lies on the task of {\em input validation}, whereas those approaches
deal with {\em parsing}, i.e. they additionally fill a data structure.

Among the previous parser generators, Gapa and Nail stand out in terms
of their safety features.  Gapa achieves a degree of safety by
generating parsers in a memory-safe language. Note that this does not
prevent runtime error, e.g., dividing by zero still remains possible.
Nail also aims at safety by providing some automated support for
filling user-defined data structure therefore reducing the risk of
errors introduced by user-defined code. In contrast, we do not rely on
any user-provided code.

Another line of work \cite{conf/esorics/DavidsonSDJ09} relies on the
use of the so-called attribute grammars, an extension of context-free
grammar that equips rules with attributes that can be accessed and
manipulated.  For the parser generator, the authors use Bison and
encode the attribute aspect of grammars through user-defined C code
annotating the grammar rules which, as we argued before, augments the
risk of errors.

\section{Conclusions and Future Work}\label{sec:conclusions}

Input validation is an important step for defending against malformed
or malicious inputs. 
In this paper we perform the first rigorous, language theoretic study of the expressiveness required for validating a number of common protocol idioms. 
We further show that input validation based on formal languages is practical and build a modular input validator for HTTP from
dependable software components such as off-the-shelf parser generators for context-free languages. 
Our experimental result shows that our approach is competitive with the state-of-the-art input validator for HTTP in terms of coverage and speed.

There are some promising avenues for extending our work. 
For instance our approach can be generalized to boolean closures of CFLs, which are known to be strictly more expressive than the finite intersection we deal with in this paper~\cite{brough2010groups}.  
Besides, our approach can be extended with a notion of state that is shared between protocol participants which will allow us to implement, e.g., stateful firewalls using our approach.

\paragraph{Acknowledgments}
We thank Juan Caballero for feedback on an early version of this paper.
This work was supported by Ram{\'o}n y Cajal grant RYC-2014-16766,
Spanish projects TIN2015-70713-R DEDETIS, TIN2015-71819-P RISCO,
TIN2012-39391-C04-01 StrongSoft, and Madrid regional project
S2013/ICE-2731 N-GREENS.

\bibliographystyle{abbrv}

\appendix

\section{Deferred Proofs}\label{sec:proofs}

\subsection{Proof of Lemma~\ref{lemma:L_cap}}

\begin{proof}[of Lemma~\ref{lemma:L_cap}]
According to Proposition \ref{prop:equivalents}, if $\strangeL$ is a finite intersection of CFLs then the set $f_w^{-1}(\strangeL)$ is semilinear, which implies that it is a finite union of linear sets. Since the set is infinite, at least one of the linear sets must be infinite. We call this set $S$.

Any element in $S$ will have the form defined in \ref{eq:LinearSetFormula} so define $\vec{x},\vec{y}\in S$ as follows:
\[\vec{x} = \left(i_x, 1,\frac{b^{i_x}-1}{b-1} \right), \;\;\;\; \vec{y} = \left(i_y,1,\frac{b^{i_y}-1}{b-1} \right) \enspace . \]

Let $δ=i_y - i_x$, we write
\[\vec{y} = \left(i_x + δ,1,\frac{b^{i_x+δ}-1}{b-1} \right)\enspace .\]
Without loss of generality, assume $δ>0$ (for otherwise swap $\vec{x}$ and $\vec{y}$). Next define
\[Δ = \vec{y} - \vec{x} = \left(δ, 0,\frac{b^{i_x}(b^δ-1)}{b-1} \right) \enspace .\]

Let \( \vec{z} = \vec{x} + 2 Δ \):
\begin{equation}\label{eq:vectorAsIncrement}
\vec{z}  = \left(i_x + 2δ,1,\frac{b^{i_x} -1 + 2b^{i_x}(b^δ-1)}{b-1} \right) = \left(i_x + 2δ,1,\frac{b^{i_x}(2b^δ-1) -1}{b-1} \right)\enspace .
\end{equation}

Since the set $S$ is linear we find that \(\vec{z} \in S\) so we can write:
\[\vec{z}= \left(i_z,1, \frac{b^{i_z}-1}{b-1} \right) \]
which means that $i_z = i_x+2δ$ obtaining:
\begin{equation}\label{eq:vectorCorrect}
\vec{z} = \left(i_x+2δ, 1,\frac{b^{i_x+2δ}-1}{b-1} \right)\enspace .
\end{equation}

Let us now derive a contradiction using \eqref{eq:vectorAsIncrement} and \eqref{eq:vectorCorrect}.
For this, we start by deriving the following equivalence:
\[b^{i_x}(2b^δ-1) = b^{i_x+2δ} \text{ if{}f } 2b^δ-1 = b^{2δ} \enspace .\]
Applying the change of variable $t=b^δ$ we obtain
\[2t-1=t^2 \text{ if{}f } t^2-2t+1 = 0 \enspace .\]
Solving the equation give us \(t = 1 \text{ if{}f } δ=0\).

So the vector $\vec{z}$ obtained basing on the linearity of the set $S$ will belong to $S$ if and only if $\vec{z}=\vec{y}=\vec{x}$.

We conclude that the set $S$ is not infinite, in fact, if it is linear then it contains only one element. Thus the set $f_w^{-1}(\strangeL)$ can not be written as a finite union of linear set so the language $\strangeL$ is not a finite intersection of CFLs.
\qed
\end{proof}

\subsection{Proof of Lemma~\ref{lemma:constraint}}
\begin{proof}[of Lemma~\ref{lemma:constraint}]
This proof is conceptually identical to the previous one but requires a slight modification in the notation used.
Now, vectors in $S$ have 7 components and the linear combination of the last five ones equals to the \emph{old} third component.

Thus, vector $\vec{x} \in S$ should be defined as:
\[\vec{x} = \left(i_x,1,n_x^{(1)},n_x^{(2)},n_x^{(3)},n_x^{(4)},n_x^{(5)}\right) \;\; \text{ with } \frac{b^{i_x}-1}{b-1}=\sum_{k=1}^5\len{α_k}\times n_x^{(k)} \enspace .\]

Having this definition, the proof follows that of Lemma~\ref{lemma:L_cap}, considering the sum of the last five components rather than their concrete values.
Assuming $L_{\rightthreetimes}$ is a finite intersection of context-free languages, there exist three vectors $\vec{x}, \vec{y},\vec{z} \in S$ such that $\vec{z} = \vec{x} + 2Δ$ where $Δ = \vec{y}-\vec{x}$.
For the sum of the last five components of vectors $\vec{z}$ and $\vec{x} + 2Δ$ to be equal we require $\vec{x}, \vec{y}$ and $\vec{z}$ to have the same first component ($δ = 0$).  

Sharing the first component implies that each linear set, $S$, relates (by $f_w^{-1}$) only with vectors with the same first component.
Therefore assuming $f_w^{-1}(L_{\rightthreetimes})$ is as a finite union of linear sets requires $L_{\rightthreetimes}$ to only contain words with a finite number of different lengths, which contradicts the fact that it is infinite.
By Proposition \ref{prop:equivalents} we conclude that $L_\rightthreetimes$ is not a finite intersection of CFLs.

\qed
\end{proof}

\subsection{Proof of Theorem~\ref{theorem:chunked}}
\begin{proof}[of Theorem~\ref{theorem:chunked}]
  Let
  $a ∈ Σ$ be such that \(a≠1\) and \(a≠\sharp\). One can easily see
  that
  \[\strangeL \{\sharp\} \stackrel{\rm\scriptscriptstyle def}{=} L_{\mathit{chunk}}^{\sharp} \cap 1^*\sharp^*a^*\sharp^* \enspace .\]

Assuming $L_{\mathit{chunk}}^{\sharp}$ is a finite intersection of CFLs then we find that so is $\strangeL \{\sharp\}$.

Let $\bar{w}=\tuple{1,\sharp,a,\sharp}$, then
\[f_{\bar{w}}^{-1}(\strangeL \{\sharp\}) = \left\{\bigl(i,1,\textstyle{\frac{b^{i}-1}{b-1}},1\bigr) \mid i\ ∈ \mathbb{N} \right\}\enspace .\]
From there, it is straightforward to extend the result of Lemma \ref{lemma:L_cap} to show that $\strangeL \{\sharp\}$ is not a finite intersection of CFLs, hence derive a contradiction.
\qed
\end{proof}

\subsection{Proof of Theorem~\ref{theorem:economychunked}}

\begin{proof}[of Theorem~\ref{theorem:economychunked}]
Recall the grammar $G_{\mathit{len}}(n)$ for the language $L_{\mathit{len}}(n)$.
The grammar $G_{\mathit{chunk}}(n)$ is defined by adding to \(G_{\mathit{len}}(n)\) a fresh start variable \(Z\) and the following two rules: \(\{ Z → S\; \sharp \; Z, Z → S \; \sharp\} \).
Clearly, \(L(G_{\mathit{chunk}}(n))=L_{\mathit{chunk}}^{\sharp}(n)\) and the size of $G_{\mathit{chunk}}(n)$ is $\mathcal{O}(n)$ since the size of \(G_{\mathit{len}}(n)\) is \(\mathcal{O}(n)\) due to Theorem~\ref{theorem:compactboundedcontentlength}

\qed
\end{proof}

\subsection{Proof of Theorem~\ref{theorem:Lww}}

\begin{proof}[of Theorem~\ref{theorem:Lww}]
We first set up languages and a GSM mapping to take proper care of the delimiter. 
Then we proceed with the proof of theorem.
Let $L^{\sharp}_{(k)}$ be the language obtained by adding the delimiter to $L_{(k)}$, that is
\[L^{\sharp}_{(k)} = \{a_1^{i_1}a_2^{i_2}\ldots a_k^{i_k}\sharp a_1^{i_1}a_2^{i_2}\ldots
a_k^{i_k} \mid i_j \geq 0 \text{ for all } j \}\enspace .\]
Now given \(k>0\) define%
\footnote{Given \(L\) over \(Σ\), let \(\overline{L}\) denotes the complement of \(L\), that is \(Σ^* \setminus L\).}
\begin{equation}\label{eq:lk}
L_{\mathit{aux}}=L^{\sharp}_{(k)} \cap \overline{\{a_i^j\sharp z \mid 1 \leq i \leq k,\; j\geq 0,\; z\in Σ^*\}} \enspace .
\end{equation}

Intuitively, \(L_{\mathit{aux}}\) is obtained from \(L^{\sharp}_{(k)}\) by filtering out the words where at most one symbol occurs to the left of the delimiter, e.g. \(a_3 a_3 \# a_1 a_2 a_3\) is such a word whereas
\(a_1 a_2 \# \) is not.
Observe that the set \(\{a_i^j\sharp z \mid 1 \leq i \leq k,\; j\geq 0,\; z\in Σ^*\}\) is regular, and so is its complement.

Next, we define the GSM mapping $π$ which intuitevely inserts the delimiter at the “right” place. 
Unless stated otherwise, when describing $π$, we assume the GSM outputs exactly what it reads.  

The GSM reads the first symbol, say $a_i$, and moves to a state \(q_i\) keeping track of that symbol. 
Notice that the first symbol need not be the delimiter following equation \eqref{eq:lk}.
Then \(π\) keeps reading the same $a_i$ symbol without leaving \(q_i\).
Upon reading some symbol \(x ≠ a_i\), \(π\) moves to a different state wherein it keeps reading symbols $x ≠ a_i$ until it reads $a_i$ again.
At that point, \(π\) has reached the middle of the word and so it outputs $\sharp$ (besides $a_i$) and keeps processing the remainder of the word.

From above, we find that \(π^{-1}(L_{\mathit{aux}})\) coincides with \(L_{(k)}\) minus the words where at most one symbol occurs, i.e. \(\{ (a_ia_i)^j \mid 1\leq i\leq k,\; j\geq 0\}\).

Now we reinstate these words and obtain the following equality:
\[L_{(k)}= π^{-1}(L_{\mathit{aux}}) \cup \{ (a_ia_i)^j \mid 1\leq i\leq k,\; j\geq 0\} \enspace .\]

We are in position to prove the statement of the theorem.

It is easy to see that the following equality holds:
\[
L_{(k)}^{\sharp}=L_{=}^{\sharp}\cap a_1^*a_2^*\ldots a_k^*\sharp a_1^*a_2^*\ldots a_k^*
\]
Now assume $L_{=}^{\sharp}$ is a finite intersection of \(m\) context-free
languages. 
Then $L_{\mathit{aux}}$ is a finite intersection of $m+2$ CFLs by equation~\ref{eq:lk}. 
Proposition~\ref{prop:closegsm} \upshape(\itshape i\upshape) shows that $π^{-1}(L_{\mathit{aux}})$ is a finite intersection of $m+2$ CFLs, hence we find that $L_{(k)}$ is an intersection of $m+2$ CFLs by Proposition~\ref{prop:closegsm} \upshape(\itshape ii\upshape).
Notice that this holds for all values of \(k\) and in particular for $k>m{+}2$ which contradicts Theorem~\ref{theorem:Km1Intersection}.
\qed
\end{proof}

\section{General equality checks}\label{sec:dynamicDelimitersAppendix}
To specify general equality checks we consider the following language over $Σ$ where \(n\) is the upper bound on the size of word to check equality for:
\[
  L_{d}(n) \stackrel{\rm\scriptscriptstyle def}{=} \{w \; x \; w \; y_1 … w \; y_k \; w \; z\; \mid \len{w}=n ∧ x∈ L_{\mathit{fm}} ∧ \bigwedge_{j} y_j ∈ L_{\mathit{mm}} ∧ z∈ L_{\mathit{lm}} \} 
\]
where $L_{\mathit{fm}}$, $L_{\mathit{lm}}$ and \(L_{\mathit{mm}}\) are CFLs.

It is routine to check that \(L_{d}(n)\) is context-free for every value of \(n\).
However, Theorem~\ref{theorem:Language_square_words} shows there is no concise specification for it.
Therefore, we study the size of specifications given as a finite intersection of CFLs.

\begin{theorem}
Assuming a fixed size grammar specification for the languages \(L_{\mathit{fm}}\), \(L_{\mathit{mm}}\) and \(L_{\mathit{lm}}\), the language $L_{d}(n)$ over fixed size alphabet \(Σ\) is a finite intersection of \(n\) CFLs, each one specified by a grammar of size $\mathcal{O}(n)$.
\end{theorem}
\begin{proof}
  Given $i \in \{1,\ldots, n\}$, define the language $L_{d_{=_i}}(n)$ over alphabet $Σ$ given by:
\begin{multline*}
  \{w_0 \; x \; w_1 \; y_1 … w_k \; y_k \; w_{k+1} \; z\; \mid \bigwedge_{\ell} (\len{w_\ell}=n ∧ (w_{\ell})_i = (w_{\ell+1})_i) \\
  ∧ x ∈ L_{\mathit{fm}} ∧ \bigwedge_{j} y_j ∈ L_{\mathit{mm}} ∧ z ∈ L_{\mathit{lm}} \} \enspace .
\end{multline*}

It is easily seen that give a word \(u\) we have \(u\in L_{d}(n)\) if{}f \(u \in L_{d_{=_i}}(n)\) for all \(i\in\{1,…,n\}\).

Let $S_{\mathit{fm}}$, $S_{\mathit{mm}}$ and $S_{\mathit{lm}}$ be the start symbols of the grammars defining languages $L_{\mathit{fm}}$, $L_{\mathit{mm}}$ and $L_{\mathit{lm}}$ respectively.
Next, define $G_{d_{=_i}}(n)$ as the grammar for \(L_{d_{=_i}}(n)\) with alphabet \(Σ\), set of variables $\{S, T, P, Q, S_{\mathit{fm}}, S_{\mathit{mm}}, S_{\mathit{lm}}\}$ and rules:
\begin{align*}
\{ S  &→ P \; S_{\mathit{fm}} \; Q \; P \; S_{\mathit{lm}} \} & \{ T &→ c \mid c ∈ Σ \} \\
\{P &→ T^{i-1}\; c\; T^{n-i}\mid c \in Σ\} & \{Q &→ P \; S_{\mathit{mm}} \; Q\} & \{Q &→ P \; S_{\mathit{mm}} \} \enspace .
\end{align*}

\(G_{d_{=_i}}(n)\) also contains all variables and rules of the grammars specifying languages $L_{\mathit{fm}}$, $L_{\mathit{mm}}$ and $L_{\mathit{lm}}$.
Because the size of the grammars defining the languages $L_{\mathit{fm}}$, $L_{\mathit{mm}}$ and $L_{\mathit{lm}}$ is fixed and independent of $n$, it is routine to check that the size of \(G_{d_{=_i}}(n)\) is $\mathcal{O}(n)$.
\qed
\end{proof}

\end{document}